\documentclass[a4paper]{article}

\usepackage[a4paper, margin=1in]{geometry}

\usepackage{hyperref}

\usepackage{times} 
\usepackage{helvet} 
\usepackage{courier} 
\usepackage{url} 
\usepackage{graphicx} 
\usepackage{caption} 

\usepackage{floatrow}
\newfloatcommand{capbtabbox}{table}[][\FBwidth]

\usepackage[toc,page]{appendix}
\usepackage{caption}
\usepackage{subcaption}
\usepackage{footnote}
\usepackage{tablefootnote}
\usepackage{multirow}
\usepackage{enumitem}
\usepackage{tikz}
\usepackage{enumitem}
\usepackage{graphicx}
\usepackage{algpseudocode}
\usepackage{mathtools}
\usepackage{listings}
\usepackage{amsmath,amssymb, mathtools, nccmath, amsthm, thmtools}
\usepackage{algorithm}

\algdef{SE}[DOWHILE]{Do}{doWhile}{\algorithmicdo}[1]{\algorithmicwhile\ #1}

\usepackage[natbibapa]{apacite}

\usepackage{subcaption}
\newcommand{\xp}{\ensuremath{x^+}}
\newcommand{\xm}{\ensuremath{x^-}}
\newcommand{\yp}{\ensuremath{y^+}}
\newcommand{\ym}{\ensuremath{y^-}}

\newcommand{\spii}{\ \ }

\newcommand{\spiv}{\ \ \ \ }
\newcommand{\spv}{\ \ \ \ \ }

\newcommand{\stexttt}[1]
  {{\tiny\tt #1}}

\renewcommand{\epsilon}{\varepsilon}

\newcommand\funcrestr[2]{{
  \left.\kern-\nulldelimiterspace 
  #1 
  \vphantom{\big|} 
  \right|_{#2} 
  }}

\newcommand{\nrd}{\textsc{nrd}}

\newcommand{\csp}{\operatorname{CSP}}
\newcommand{\rcci}[1]{\textsc{RCC-}\ensuremath{#1}}
\newcommand{\bigoh}{\mathcal{O}}

\newtheorem{theorem}{Theorem}
\newtheorem{definition}[theorem]{Definition}

\newtheorem{lem}[theorem]{Lemma}

\newtheorem{example}[theorem]{Example}
\newtheorem{col}[theorem]{Corollary}
\newtheorem{obs}[theorem]{Observation}

\newtheorem*{claim*}{Claim}

\newlist{thmlist}{enumerate}{1}
\setlist[thmlist]{label=(\roman{thmlisti}), ref=\thetheorem.(\roman{thmlisti}),noitemsep, align=right}

\title{Towards Single Exponential Time for Temporal and Spatial Reasoning: A Study via Redundancy and Dynamic Programming
}

\author{Victor Lagerkvist,
    Johanna Groven,
    Leif Eriksson
}

\date{
    Linköping University,
    581 83 Linköping, Sweden\\
    victor.lagerkvist@liu.se, johanna.groven@liu.se,
    leif.eriksson@liu.se \\[2ex]
}

\begin{document}

\maketitle

\begin{abstract}
\noindent The {\em region connection calculus} (\textsc{RCC}) and {\em Allen's interval algebra} (\textsc{IA}) are two well-known NP-hard spatial-temporal qualitative reasoning problems. 
They are solvable in $2^{\bigoh(n \log n)}$ time,
where $n$ is the number of variables, and $\textsc{IA}$ is additionally known to be solvable in $o(n)^n$ time. However, no improvement over exhaustive search is known for $\textsc{RCC}$, and if they are also solvable in single exponential time $2^{\bigoh(n)}$ is unknown. We investigate multiple avenues towards reaching such bounds. First, we show that branching is insufficient since there are too many {\em non-redundant} constraints. Concretely, we classify the maximum number of non-redundant constraints in \textsc{RCC} and \textsc{IA}. Algorithmically, we make two significant contributions based on dynamic programming (DP). The first algorithm runs in $4^n$ time and is applicable to a non-trivial, NP-hard fragment of \textsc{IA}, which includes the well-known {\em interval graph sandwich problem} of~\cite{IGSPpaper}.
For the richer \textsc{RCC} problem with 8 basic relations we use a more sophisticated approach which asymptotically matches the $o(n)^n$ bound for $\textsc{IA}$.
\end{abstract}

\section{Introduction}

{\em Infinite} structures are a cornerstone in the current artificial intelligence (AI) revolution, ranging from neural networks (which typically approximate functions over continuous structures) to classical formalisms such as {\em qualitative reasoning}.
Here, the basic objects are regions in e.g.\ space/time and the task is to determine whether a given configuration (with uncertainties) is consistent. These problems have rich expressive power, are generally NP-complete, and can be formulated as infinite-domain {\em constraint satisfaction problems} (CSPs). While they admit large tractable classes~\cite{dylla} their precise exponential complexity is poorly understood. 
Many finite-domain problems  can be solved in $2^{\bigoh(n)}$ time
by enumerating functions from $n$ objects (e.g., variables) to the fixed domain. 
This can often be matched with asymptotically tight lower bounds $2^{o(n)}$ (unless the {\em exponential-time hypothesis} fails~\cite{impagliazzo2001}).
For infinite domains exhaustive enumeration is more complicated, and for qualitative reasoning this only gives $2^{\bigoh(n^2)}$ time ($n$ is the number of variables) which is often improvable to $2^{\bigoh(n \log n)}$ time. Clearly, this is  {\em much} worse than $2^{\bigoh(n)}$, and the best lower bounds only rule out $2^{o(n)}$ algorithms~\cite{IncompRelation}.

In this paper we are interested in narrowing the gap between infinite and finite structures: when are infinite-domain CSPs solvable in $2^{\bigoh(n)}$ time? We study this via qualitative reasoning problems since they are well understood from the angle of classical complexity theory, model theory, and algebra~\cite{DBLP:journals/jair/BodirskyJ17}, and thus form a suitable microcosmos to tackle the $2^{\bigoh(n)}$ versus $2^{\bigoh(n \log n)}$ question. We concentrate on the two most well-known formalisms: {\em Allen's interval algebra} (\textsc{IA})~\cite{AllenAlgebra} and the {\em region connection calculus} (\textsc{RCC})~\cite{RCCIntroduction}.
Beating exhaustive search is possible for \textsc{IA} with the first breakthrough being a $\bigoh((1.0615n)^{n})$ time dynamic programming (DP) algorithm~\cite{OriginalPartitioningAlg}, subsequently improved to $o(n)^n$~\cite{lagerkvist2023a}.
For \rcci{5} and \rcci{8} the situation is worse: no improvement is known over enumerating certain orders, which solves \rcci{8} in $\bigoh((0.531n)^n)$ time and \rcci{5} in $\bigoh((0.368n)^n)$ time~\cite{IncompRelation}.

We consider two orthogonal approaches. First, if we could reduce the number of constraints $m$ from $\bigoh(n^2)$ towards $\bigoh(n)$ then $2^{\bigoh(m)}$ branching would suddenly be theoretically interesting. Thus, we want to identify constraints that are {\em redundant} in the sense that their removal does not affect the solution space (constructing {\em prime} instances). This is a classical problem in qualitative reasoning~\cite{dylla} and prime instances can often be computed quickly~\cite{DBLP:conf/ijcai/SioutisLC15,DBLP:conf/iske/HuGL00S21}. 
Curiously, while there are complexity results for identifying redundant constraints~\cite{LI201551}, nothing rigorous is known about the maximal number of non-redundant constraints.
To address this we describe the maximal $n$-variable prime instance of a given relation $R$ ($\nrd_{\{R\}}(n)$), initially studied due to its connection to VC-dimension and learnability~\cite{DBLP:conf/aaai/BessiereCK20}. We provide a complete classification of $\nrd_{\{R\}}(n)$ for all basic relations, and several natural macro relations in \textsc{IA} and \textsc{RCC}: $\nrd_{\{R\}}(n) \in \Theta(n^2)$ for most relations, while $\nrd_{\{R\}}(n) \in \Theta(n)$ is only reached when $R$ is the equality or the {\em meets} relation in \textsc{IA}. Thus, improvements to $n^2$ can be made but prime instances with only $\bigoh(n)$ constraints are in general impossible.

Our second contribution is algorithmic and here we primarily explore how DP, which is not sensitive to the number of constraints, can be exploited to get closer to $2^{\bigoh(n)}$. First, we consider a non-trivial (but NP-hard) fragment of \textsc{IA} based around strict partial orders $p$, $p^{-1}$ and an incomparability relation $\cap$. In particular this problem captures the well-known, NP-hard \textit{interval graph sandwich problem} (\textsc{IGSP}) 
with applications in e.g.\ molecular biology~\cite{IGSPpaper}. We present a DP algorithm that solves \textsc{IGSP} in $\bigoh(4^n)$ time (ignoring polynomial factors). 
The algorithm is based on partitioning the endpoints of intervals into two classes $V_<, V_>$ where endpoints gradually are moved from the right to the left without introducing inconsistencies. This is inspired by the DP algorithm by~\cite{OriginalPartitioningAlg} for the $\csp$ problem with first-order definable relations over $(\mathbb{Q}, <)$ of arity at most 3, but the generalization to interval relations is not without complications. This algorithm cannot (to the best of our knowledge) be generalized to \rcci{8} while preserving the single exponential running time, and, given that {\em no} improvement is known over exhaustive search we first concentrate on matching the $o(n)^n$ bound for \textsc{IA}~\cite{lagerkvist2023a}. 
A straightforward DP approach for \rcci{8}, such as iteratively building orderings, quickly runs into issues because it is unclear which intermediate states to build upon when multiple valid ones exist. Unlike our algorithm for IGSP , early choices in \rcci{8} can significantly impact later stages. We resolve this by storing all necessary states and the novel idea of comparing them based on what is {\em not yet solved}.

Hence, we not only explain {\em why} branching algorithms cannot give $2^{\bigoh(n)}$ bounds for qualitative reasoning problems but also identify the first fragment solvable in $2^{\bigoh(n)}$ time as well as giving the first significant improvement for \rcci{8}.

\section{Preliminaries}

We sometimes use the notation $2^{\bigoh(n)}$ to express running times $2^{f(n)}$ for $f \in \bigoh(n)$. This notation allows one to hide factors polynomial in the input size but it is sometimes also convenient to express running times dominated by $f(n)^n$ for reasonably small $f$, where the constants may be important. To suppress polynomial factors in the input size we write $\bigoh^*(f(n)^n)$. We assume that the reader is familiar with basic notions such as (strict) partial/total orders, reflexivity, symmetry, and transitivity. For a strict total or partial order $O = (S, <_O)$ we write $\leq_O$ for the non-strict variant of $<_O$, and vice versa if we are given a non-strict order $(S, \leq_O)$. Similarly, we write $=_O$ for the equality relation induced by a given order $O$. For a set $S$, we call a partition $(S_1,\ldots,S_K)$ that partitions $S$ into $k$ different sets $S_1,\ldots,S_k \subseteq S$ a $k$-partition of $S$.
    
\subsection{Constraint Satisfaction Problems}

In this paper we are exclusively concerned with binary relations over a fixed (but infinite) domain $\mathcal{D}$. A set of binary relations $\Gamma$ over $\mathcal{D}$ is called a {\em constraint language}.
If $\Gamma = \{R_1, \ldots, R_m\}$ is finite and (1) $\bigcup \Gamma = \mathcal{D}^2$ ({\em jointly exhaustive}), (2) $R_i \cap R_j = \emptyset$ for all $i \neq j$ ({\em pairwise disjoint}), (3) ${\sf Eq}_{\mathcal{D}} = \{(x, x) \mid x \in \mathcal{D}\} \in \Gamma$, (4) if $R_i \in \Gamma$ then $\Gamma$ also contains the {\em inverse} $R_i^{-1} = \{(x,y) \mid (y,x) \in R_i\}$, then $\Gamma$ is said to be a {\em partition scheme}. Partition schemes are the predominant way to define qualitative reasoning problems (cf. ~\cite{dylla}).
For a set $\Gamma' \subseteq \Gamma$ of relations we define its {\em macro-relation} as $R^U_{\Gamma'} \coloneqq \bigcup_{R \in \Gamma'} R$. We define the closure of $\Gamma$ under macro-relations as $\Gamma^U \coloneqq \bigcup_{\Gamma' \subseteq \Gamma} \{R^U_{\Gamma'}\}$. 

In the {\em constraint satisfaction problem} over a constraint language $\Gamma$ ($\csp(\Gamma)$) over a domain $\mathcal{D}$ we are given a set of variables $V$ and a set of constraints $C$ of the form $R(x,y)$ where $R \in \Gamma$ and $x,y \in V$, and want to know whether there exists a {\em satisfying assignment} $f \colon V \to \mathcal{D}$ such that $(f(x), f(y)) \in R$ for every constraint $R(x,y) \in C$. 
We are specifically interested in the case when $\Gamma$ is a partition scheme where $\csp(\Gamma)$ is in P, and then want to solve the richer (and generally NP-complete) problem $\csp(\Gamma^U)$. For these problems, if for $x,y \in V$, there are multiple constraints $R_S^U(x,y) \in C$ and $R_{S'}^U(x,y) \in C$, we can always polynomially reduce to an instance $(V,C')$ where we replace $R_S^U(x,y)$ and $R_{S'}^U(x,y)$ by $R_{S \cap S'}^U(x,y)$ in $C'$, which is still a $\csp(\Gamma^U)$ constraint by definition. We thus often assume that for each pair $x,y \in V$, there is at most one constraint $R_S^U(x,y) \in C$. This allows us to also
use the {\em relational network} representation where we represent an instance $(V,C)$ of $\csp(\Gamma^U)$ by a total function $f$ defined by $f(x,y) = S$ if $R^U_S(x,y) \in C$ and $f(x,y) = \Gamma$ otherwise. Note that if $|f(x,y)| = 1$ for each pair $x,y \in V$ where $f(x,y) \neq \Gamma$ then the instance can be viewed as a $\csp(\Gamma)$ instance. When this distinction is important then we sometimes refer to the former as a {\em multi relational network} and the latter as a {\em simple relational network}, or an {\em instantiation}. A simple relational network $f$ is said to be {\em consistent} with a $\csp(\Gamma^U)$ instance $(V,C)$ if $f(x,y) \in \{R_1, \ldots, R_m\}$ for each constraint $R^U_{\{R_1, \ldots, R_m\}}(x,y)$. We sometimes view a relational network as a graph $(V,E_C)$ with labeled edges $E_C = \{(x,y,f(x,y)) \mid x,y \in V, f(x,y) \neq \Gamma\}$.

We also need one important concept from model theory. First, given a partition scheme $\Gamma$ over $\mathcal{D}$ the {\em substructure} induced by $S \subseteq \mathcal{D}$ is the set of relations of the form $R \cap S^2$ for $R \in \Gamma$. Second, given two substructures $\mathcal{A}$ (over $A \subseteq \mathcal{D}$) and $\mathcal{B}$ (over $B \subseteq \mathcal{D}$) an {\em isomorphism} is a bijective function $h \colon A \to B$ where (1)  $(h(a_1), h(a_2)) \in R \cap B^2$ for all $(a_1, a_2) \in R \cap A^2$, for every $R \in \Gamma$, and (2) the inverse $h^{-1}$ satisfies $(h^{-1}(b_1), h^{-1}(b_2)) \in R \cap A^2$ for all $(b_1, b_2) \in R \cap B^2$, for every $R \in \Gamma$. An {\em automorphism} of $\Gamma$ is an isomorphism from $\Gamma$ to $\Gamma$. We then say that a partition scheme $\mathcal{A}$ is {\em homogeneous} if any isomorphism between two finite substructures can be extended to an automorphism of $\mathcal{A}$. For additional details, see e.g. the book by Bodirsky~\citeyearpar{Bodirsky_2021}.

\subsection{Qualitative Reasoning Problems}
Let $\mathbb{I}_\mathbb{Q} = \{[a,b] \mid a,b \in \mathbb{Q}, a < b\}$ be the set of all nonempty closed intervals on the rational line.~\cite{AllenAlgebra} introduced a partition scheme $\mathfrak{A}_{13} = \{p,m,o,s,f,d,e,d^{-1},f^{-1},s^{-1},o^{-1}, m^{-1},p^{-1}\}$ over $\mathbb{I}_\mathbb{Q}$, where $p, m, o, s, f, d, e$ stands for {\em precedes}, {\em meets}, {\em overlaps}, {\em starts}, {\em finishes}, {\em during}, and {\em equals}. All these relations can  be formally defined by their endpoints, e.g., for two intervals $x = [x^-,x^+]$ and  $y =  [y^-,y^+]$ then $x p y$ iff $x^+ < y^-$, $x^- < x^+$, and $y^- < y^+$. 

The thirteen relations and their equivalence to endpoint relations can be defined as seen in Table~\ref{tb:allen}. Note that the representation based on $\mathbb{I}_\mathbb{Q}$ for $\textsc{IA}$ chosen here is a homogeneous representation~\cite{IAHomogenousness}.

Fragments using macro-relations are also often studied. For this purpose, we here consider $\alpha \coloneqq R^U_{\{m,o\}},$ $ \alpha^{-1} \coloneqq R^U_{\{m^{-1},o^{-1}\}}, \subset \coloneqq R^U_{\{s,f,d\}}, \subset^{-1} \coloneqq R^U_{\{s^{-1},f^{-1},d^{-1}\}}$ and $\cap \coloneqq R^U_{\mathfrak{A}_{13} \setminus \{p,p^{-1}\}}$. These relations induce the language 
    $\mathfrak{A}_3 \coloneqq \{p, \cap, p^{-1} \}$ and
    $\mathfrak{A}_7 \coloneqq \{p, p^{-1}, \alpha, \alpha^{-1}, \subset, \subset^{-1},e \}$
which induces NP-hard problems $\csp(\mathfrak{A}^U_3)$ and $\csp(\mathfrak{A}^U_7)$~\cite{IGSPpaper}. We remark that all known maximal tractable subclasses of $\mathfrak{A}^U_{13}$ have been identified~\cite{dylla}.

\begin{figure}[b!]
    \begin{floatrow}
    \centering
      \capbtabbox{\small
    \begin{tabular}{|ll|c|l|}\hline
       Basic relation        &          & Example & Endpoints\\
\hline\hline
       $x$ precedes       $y$ & ${\sf p}$  & \stexttt{xxx\spv}    & $\xp<\ym$
 \\ \cline{1-2}
       $y$ preceded by         $x$ & ${\sf p}^{-1}$  & \stexttt{\spv yyy}   & \\
\hline
       $x$ meets         $y$ & ${\sf m}$   & \stexttt{xxxx\spiv}  & $\xp=\ym$
\\ \cline{1-2}
       $y$ met-by        $x$ & ${\sf m}^{-1}$  & \stexttt{\spiv yyyy} & \\
\hline
       $x$ overlaps      $y$ & ${\sf o}$   & \stexttt{xxxx\spii}  & $\xm<\ym<\xp$, \\ \cline{1-2}
       $y$ overl.-by $x$ & ${\sf o}^{-1}$  & \stexttt{\spii yyyy} & $\xp<\yp$
\\ \hline
       $x$ during        $y$ & ${\sf d}$   & \stexttt{\spii xxx\spii} & $\xm>\ym$, \\ \cline{1-2}
       $y$ includes      $x$ & ${\sf d}^{-1}$  & \stexttt{yyyyyyy}   & $\xp<\yp$ \\
\hline
       $x$ starts        $y$ & ${\sf s}$   & \stexttt{xxx\spiv}  & $\xm=\ym$,
\\ \cline{1-2}
       $y$ started by    $x$ & ${\sf s}^{-1}$  & \stexttt{yyyyyyy}   & $\xp<\yp$ \\
\hline
       $x$ finishes      $y$ & ${\sf f}$   & \stexttt{\spiv xxx} & $\xp=\yp$,
\\ \cline{1-2}
       $y$ finished by   $x$ & ${\sf f}^{-1}$  & \stexttt{yyyyyyy}   & $\xm>\ym$ \\
\hline
       $x$ equals        $y$ & $e$ & \stexttt{xxxx}      & $\xm=\ym$, \\
                                   &          & \stexttt{yyyy}      & $\xp=\yp$ \\
\hline
    \end{tabular}
    }{\caption{The thirteen basic relations in Allen's interval algebra for intervals $x=[x^-,x^+]$ and $y=[y^-,y^+]$. The endpoint relations $x^- < x^+$ and $y^- < y^+$ that are valid for
  all relations have been omitted.}
    \label{tb:allen}}
\ffigbox{
  \centering
  {\setlength{\tabcolsep}{1pt}
   \begin{tabular}{cccc}
     \begin{subfigure}[b]{0.235\columnwidth}
       \centering
       \begin{tikzpicture}[scale=0.75]
         \draw (0,0) circle (1);
         \draw[dashed] (0,0) circle (1);
         \node at (-0.33,0) {$X$};
         \node at (0.33,0)  {$Y$};
       \end{tikzpicture}
       \caption{\small ${\sf EQ}(X,Y)$}
     \end{subfigure} &
     \begin{subfigure}[b]{0.235\columnwidth}
       \centering
       \begin{tikzpicture}[scale=0.5]
         \draw (0,0.75) circle (1);
         \draw[dashed] (0,-0.75) circle (1);
         \node at (0,0.75)  {$X$};
         \node at (0,-0.75) {$Y$};
       \end{tikzpicture}
       \caption{\small ${\sf PO}(X,Y)$}
     \end{subfigure} &
     \begin{subfigure}[b]{0.235\columnwidth}
       \centering
       \begin{tikzpicture}[scale=0.5]
         \draw (0,1) circle (1);
         \node at (0,1)     {$X$};
         \draw[dashed] (0,-1) circle (1);
         \node at (0,-1)    {$Y$};
       \end{tikzpicture}
       \caption{\small ${\sf EC}(X,Y)$}
     \end{subfigure} &
     \begin{subfigure}[b]{0.235\columnwidth}
       \centering
       \begin{tikzpicture}[scale=0.5]
         \draw (0,1.1) circle (1);
         \node at (0,1.1)   {$X$};
         \draw[dashed] (0,-1.1) circle (1);
         \node at (0,-1.1)  {$Y$};
       \end{tikzpicture}
       \caption{\small ${\sf DC}(X,Y)$}
     \end{subfigure}
     \\[-2pt] 
     \begin{subfigure}[b]{0.235\columnwidth}
       \centering
       \begin{tikzpicture}[scale=0.7]
         \draw (-0.5,0) circle (0.5);
         \node at (-0.5,0) {$X$};
         \draw[dashed] (0,0) circle (1);
         \node at (0.5,0)  {$Y$};
       \end{tikzpicture}
       \caption{\small ${\sf TPP}(X,Y)$}
     \end{subfigure} &
     \begin{subfigure}[b]{0.235\columnwidth}
       \centering
       \begin{tikzpicture}[scale=0.7]
         \draw[dashed] (-0.5,0) circle (0.5);
         \node at (-0.5,0) {$Y$};
         \draw (0,0) circle (1);
         \node at (0.5,0)  {$X$};
       \end{tikzpicture}
       \caption{\small ${\sf TPP}^{-1}(X,Y)$}
     \end{subfigure} &
     \begin{subfigure}[b]{0.235\columnwidth}
       \centering
       \begin{tikzpicture}[scale=0.7]
         \draw (-0.3,0) circle (0.5);
         \node at (-0.3,0) {$X$};
         \draw[dashed] (0,0) circle (1);
         \node at (0.5,0)  {$Y$};
       \end{tikzpicture}
       \caption{\small ${\sf NTPP}(X,Y)$}
     \end{subfigure} &
     \begin{subfigure}[b]{0.235\columnwidth}
       \centering
       \begin{tikzpicture}[scale=0.7]
         \draw[dashed] (-0.3,0) circle (0.5);
         \node at (-0.3,0) {$Y$};
         \draw (0,0) circle (1);
         \node at (0.5,0)  {$X$};
       \end{tikzpicture}
       \caption{\small ${\sf NTPP}^{-1}(X,Y)$}
     \end{subfigure}
   \end{tabular}
  }
}{
  \caption{Illustration of the basic relations of \rcci{8} with two-dimensional disks.}
  \label{fig:rcc8}
}
\end{floatrow}
\end{figure}

\cite{RCCIntroduction} introduced a set of eight relations $\rcci{8} = \{EQ, \allowbreak PO, \allowbreak EC, \allowbreak DC, \allowbreak TPP, \allowbreak TPP^{-1}, \allowbreak NTPP, \allowbreak NTPP^{-1}\}$ (see Figure~\ref{fig:rcc8}).
We also consider the macro-relations $PP \coloneqq R^U_{\{TPP,NTPP\}}, PP^{-1} \coloneqq R^U_{\{TPP^{-1},NTPP^{-1}\}}$ and $DR \coloneqq R^U_{\{EC,DC\}}$. These relations induce the fragment $\rcci{5} = \{EQ, PO, DR,$ $PP, PP^{-1}\}$ of $\textsc{RCC}$ that also forms a partition scheme. 
$\csp(\rcci{5}^U)$ and $\csp(\rcci{8}^U)$ are both NP-hard 
and all maximal tractable subclasses have been identified (cf.~\cite{dylla}). 

We observe that a homogeneous representation of \textsc{RCC} is well-known to exist~\cite{RCCHomogenousness} and we assume to use this representation and domain throughout this paper. We emphasize that the precise technical details of this representation are not of importance to us, as we work from a relational perspective in this paper. For additional details and precise technical definitions of all \rcci{8} relations beyond the example in the main paper, we refer to the above article by Bodirksy and Wölfl~\citeyearpar{RCCHomogenousness}.

\section{Redundancy} \label{sec:kernandred}

If $\Gamma$ is a partition scheme where $\csp(\Gamma)$ is solvable in polynomial time then an arbitrary instance $(V,C)$, $|V| = n$, $|C| = m$, of $\csp(\Gamma^{U})$ can be solved in $2^{\bigoh(m)} \cdot (n + m)^{\bigoh(1)}$ time by exhaustive branching~\cite{IncompRelation}. Specifically, for $\csp(\mathfrak{A}_{13}^{U})$ we obtain the bound $2^{m \cdot \log_2 12} \cdot (n+m)^{\bigoh(1)}$, and for $\csp(\rcci{8}^U)$ the moderately better bound $2^{m \cdot \log_2 7} \cdot (n+m)^{\bigoh(1)}$. Unfortunately, since we may have a quadratic number of constraints, these bounds are asymptotically {\em worse} than baseline upper bounds of $2^{\bigoh(n \log n)} \cdot (n+m)^{\bigoh(1)}$ which can be obtained by enumerating orders~\cite{3ColtoAllensReduction,IncompRelation}. Hence, in this section we investigate methods for decreasing $m$, with a particular focus on 
 identification of {\em redundant} constraints. We give strong evidence that this method is not likely to bring $2^{\bigoh(m)}$ closer to $2^{\bigoh(n)}$. We stress that analyzing expected running times of branching algorithms equipped with sophisticated heuristics on such instances is a significantly harder problem.

Let $\Gamma$ be a finite constraint language over a domain $\mathcal{D}$, and let $(V,C)$ be a $\csp(\Gamma)$ instance (where $C$ is a set of constraints, i.e., not the network representation).
We say that $c \in C$ is {\em non-redundant} if there exists a satisfying assignment $I: V \rightarrow \mathcal{D}$ to $(V, C \setminus \{c\})$ that does not satisfy $(V, C)$, and we say that $(V, C)$ is {\em non-redundant}, or {\em prime}, if every constraint in $C$ is non-redundant. For each $n \geq 1$ define $\nrd_{\Gamma}(n)$ over the set of $n$-variable, non-redundant instances $\mathcal{I}_\Gamma(n)$ of $\csp(\Gamma)$ as 
$\nrd_{\Gamma}(n) \coloneqq \max \{|C| \mid (V,C) \in \mathcal{I}_\Gamma(n)\}$.
This function is well-defined since there exists a finite number of $n$-variable $\csp(\Gamma)$ instances, and, furthermore, we have $\nrd_{\Gamma}(n) \in \bigoh(n^2)$, as well as $\max_{R \in \Gamma}\nrd_{\{R\}}(n) \leq \nrd_{\Gamma}(n) \leq \Sigma_{R \in \Gamma} \nrd_{\{R\}}(n)$~\cite{DBLP:conf/aaai/BessiereCK20}.

We now continue classifying the non-redundancy by identifying $n$-variable non-redundant instances of all relations in $\mathfrak{A}_i$ and \rcci{j} for $i \in \{3,7,13\}$ and $j \in \{5,8\}$. In most cases we can show that these are indeed maximal. 

\begin{table}[b!]
    \begin{center}
      {\small
    \begin{tabular}{|c|c|c|c|c|c|}\hline
       \multirow{2}{*}{$\nrd_{\{R\}}(n)$} & \multicolumn{3}{|c|}{$\textsc{IA}$} & \multicolumn{2}{|c|}{$\textsc{RCC}$} \\
       \cline{2-6}
               &   $\mathfrak{A}_3$ & $\mathfrak{A}_7$ & $\mathfrak{A}_{13}$   & $\rcci{5}$ & $\rcci{8}$ \\
       \hline\hline
       \multirow{2}{*}{$\frac{n(n-1)}{2}$} &   \multirow{2}{*}{$-$} & \multirow{2}{*}{$-$} & \multirow{2}{*}{$-$} & $DR,$ & $EC,DC,$ \\
       &    & &  & $PO$ & $PO$ \\
       \hline
       $\Omega(\lceil\frac{n}{2}\rceil \cdot \lfloor\frac{n}{2}\rfloor)$ & $\cap$ & $\alpha$ & $o$ & $-$ & $TPP$ \\
       \hline
       \multirow{2}{*}{$\lceil\frac{n}{2}\rceil \cdot \lfloor\frac{n}{2}\rfloor$} &   \multirow{2}{*}{$p$} & \multirow{2}{*}{$p,\subset$} & $p,s,$ & \multirow{2}{*}{$PP$} & \multirow{2}{*}{$NTPP$} \\
        &   & & $f,d$ &  & \\
       \hline
       $2n-4$ &   $-$ & $-$ & $m$ & $-$  & $-$ \\
       \hline
       $n-1$ &   $-$ & $e$ & $e$ & $EQ$  & $EQ$ \\
       \hline
    \end{tabular}}
\caption{The non-redundancy $\nrd_{\{R\}}(n)$ for $n\geq3$ for all \textsc{IA} relations $R \in \mathfrak{A}_{i}$ for $i \in \{3,7,13\}$ and for all \textsc{RCC} relations $R \in \rcci{i}$ over $\mathbb{R}^d$ for $i \in \{5,8\}$. The leftmost column shows the $\nrd_{\{R\}}(n)$ value for all relations $R$ in that the corresponding rows. Each other column classifies $\nrd_{\{R\}}(n)$ for all relations of one of these algebras. Converse relations are omitted as they have the same $\nrd_{\{R\}}(n)$ as their converse. The symbol '-' marks brackets where an algebra has no applicable relation $R$.}
\label{tb:redundancy}
    \end{center}
\end{table}
\begin{theorem} \label{theo:mainred}
    The non-redundancy $\nrd_{\{R\}}(n)$ of all relations $R \in \mathfrak{A}_{i}$ for $i \in \{3,7,13\}$ and of all \textsc{RCC} relations $R \in \rcci{i}$ for $i \in \{5,8\}$ can be classified as seen in Table~\ref{tb:redundancy} for $n\geq 3$.
\end{theorem}

To proof Theorem~\ref{theo:mainred}, we first show a more general result, yielding bounds for $\nrd_{\{R\}}(n)$ for well-behaved relations $R$. For a relation $R \subseteq \mathcal{D}^2$ over domain $\mathcal{D}$, 
we say that $R$ has {\em infinite height} if there exists a series $(x_i)_{i \in \mathbb{N}}$ of distinct elements in $\mathcal{D}$ such that $(x_i,x_{i+1}) \in R$ for all $i \in \mathbb{N}$. For relations that share this property and are well-behaved in some other fashions, we infer the following general result.
\begin{lem} \label{lem:redundancygeneralbounds}
    Let $\Gamma$ be a homogenous partition scheme over domain $\mathcal{D}$. Let $R \subset \mathcal{D}^2, R \in \Gamma$ be a binary relation with infinite height that is symmetric or antisymmetric and reflexive or antireflexive. Then, for $n>2$, we have 
    \[n-1 \leq \nrd_{\{R\}}(n) \leq \frac{n(n-1)}{2}.\]
\end{lem}
    \begin{proof}
    For the lower bound, let $(V,C)$ be an instance of $\csp(\{R\})$ such that $(V,E_C)$ is a spanning tree. We claim that $(V,C)$ is a non-redundant instance showing $\nrd_{\{R\}}(n) \geq n-1$. We have $|C| \geq n-1$ by construction, so it remains to show that $(V,C)$ is non-redundant. To see this, we first show that $(V,C)$ is satisfiable by constructing a satisfying assignment $\sigma$. As $R\neq \emptyset$, there must be a $(a,b) \in R$. As $\Gamma$ is a partition scheme, we have ${\sf Eq}_{\mathcal{D}} = \{(x,x) \mid x \in \mathcal{D}\} \in \Gamma$. Then, for any $c \in \mathcal{D}$, we have $(a,a),(c,c) \in {\sf Eq}_{\mathcal{D}}$. By homogeneity, we conclude that there exists an isomorphism between $a$ and $c$ that can be extended to an automorphism $\alpha$ on $\Gamma$ such that $\alpha(a)=c$ and thus $(c,\alpha(b)) \in R$. As $c$ was chosen arbitrarily, we conclude that for every $a \in \mathcal{D}$, there exists an $b \in \mathcal{D}$ such that $(a,b) \in R$. We can fix a root $v \in V$ and an assignment $\sigma(v) = a$ of $v$ to $a \in \mathcal{D}$. We know there must exist $(a,b)\in R$ for some $b \in \mathcal{D}$. For all children $w \in V$ of $v$ in $(V,E_C)$ we can thus safely fix $\sigma(w) = b$, satisfying all constraints $R(v,w) \in C$. Applying this recursively, we can construct an assignment for the entire instance, as $(V,E_C)$ is a tree.
    Take now any one constraint $R(v,w) \in C$. We again construct an assignment $\sigma$ that satisfies all constraints but $R(v,w)$, proving its non-redundancy. As $R \neq \mathcal{D}^2$, there is $(a,b) \in \mathcal{D}^2 \setminus R$. Fix $\sigma(v) = a$ and $\sigma(w) = b$. Consider now the instance $(V,C\setminus\{R(v,w)\})$ and its corresponding graph $(V,E_{C\setminus \{R(v,w)\}})$. Observe that as $(V,E_C)$ is a spanning tree, $(V,E_{C\setminus \{R(v,w)\}})$ consists of two connected components $(V_v,E_v),(V_w,E_w)$ that are spanning trees rooted at $v$ and $w$ respectively. As before we may then extend $\sigma$ such that it satisfies all contraints in these connected components, but not $R(v,w)$, concluding the proof.

    For the upper bound, assume there is a a non-redundant instance $(V,C)$ of $\csp(\{R\})$ with $|C| > \frac{n(n-1)}{2}$ for $|V| = n > 2$. By a simple counting argument, $C$ must contain a constraint $R(x,x) \in C$ for $x \in V$ or it must contain constraints $R(x,y) \in C$ and $R(y,x) \in C$ for some $x,y \in V$. Assume there is $x \in V$ with $R(x,x) \in C$. If $R$ is antireflexive, $(V,C)$ is unsatisfiable and all constraints but $R(x,x)$ are redundant. Thus, $(V,C)$ is not a non-redundant instance of $\csp(\{R\})$ (as $n>2$). If $R$ is reflexive, every solution to $(V,C)$ also satisfies $(V,C\setminus\{R(x,x)\})$ and $(V,C)$ is not a non-redundant instance of $\csp(\{R\})$. Next, assume there is $x,y \in V$ such that we have $R(x,y) \in C$ and $R(y,x) \in C$. If $R$ is antisymmetric, $(V,C)$ is unsatisfiable and all constraints but $R(x,y)$ and $R(y,x)$ are non-redundant (as $n>2$). Thus, $(V,C)$ is not a non-redundant instance of $\csp(\{R\})$. If $R$ is symmetric, every solution to $(V,C)$ also satisfies $(V,C\setminus\{R(x,y)\})$ and $(V,C)$ is not a non-redundant instance of $\csp(\{R\})$. We conclude that no non-redundant instance $(V,C)$ of $\csp(\{R\})$ with $|C| > \frac{n(n-1)}{2}$ can exist for $|V| = n > 2$. Thus, we also have  $\nrd_{\{R\}}(n) \leq \frac{n(n-1)}{2}$ and have proven our claim.
    \end{proof}

    With this lemma in place, we can now present the proof for Theorem~\ref{theo:mainred}.
    
    \begin{proof}[Proof of Theorem~\ref{theo:mainred}]
    The settings of \textsc{IA} and \textsc{RCC} fulfill the conditions of Lemma~\ref{lem:redundancygeneralbounds}.
    Additionally, as all considered relations are reflexive or antireflexive as well as symmetric, antisymmetric or both and have infinite height, we have $n-1 \leq \nrd_{\{R\}}(n) \leq \frac{n(n-1)}{2}$ by Lemma~\ref{lem:redundancygeneralbounds}.

For relations $R$ for which we claim $\nrd_{\{R\}}(n) = \frac{n(n-1)}{2}$, consider the instance $(V,C)$ with $V = \{x_1,\ldots,x_n\}$ and $C = \{R(x_i,x_j) \mid i,j \in [n], i < j \}$. We claim that this instance is a non-redundant instance for $R \in \{DR,EC,DC,PO\}$. The proofs for all choices of $R$ are analogous in construction, so we show this here as an example for $R = \{DC\}$, for which this is simplest. 
Let $DC(x_k,x_\ell) \in C$ be an arbitrary constraint. Consider the assignment $\sigma$ that assigns each $x_i, i \in [n]$ a distinct region such that all pairs of regions $\sigma(x_i),\sigma(x_j)$ are disjoint. The precise shape of the regions is not important for this proof. Then $\sigma$ is clearly a satisfying assignment for $(V,C)$ and $(V,C\setminus\{DC(x_k,x_\ell)\})$. Now consider the assignment $\sigma'$ that behaves just like $\sigma$ but with $\sigma'(x_k) = \sigma(x_\ell)$. Clearly, $\sigma'$ is still a satisfying assignment for $(V,C\setminus\{DC(x_k,x_\ell)\})$ but not for $(V,C)$, as all regions $\sigma(x_i)$ are pairwise disjoint but $\sigma(x_k),\sigma(x_\ell)$, which are now identical. By definition, $DC(x_k,x_\ell)$ is thus a non-redundant constraint in $C$. Since we assumed $DC(x_k,x_\ell)$ to be an arbitrary constraint in $C$, all constraints in $C$ are redundant and $(V,C)$ is a non-redundant, $n$-variable instance of $\csp(\{DC\})$. Analogously, for the other relations $R\in \{DR,EC,PO\}$, we choose a satisfying assignment where they pairwise fulfill the wanted property, e.g. being externally connected. Then the same principle applies that we cannot differentiate regions $x_k,x_\ell$ without the constraint $R(x_k,x_\ell)$ and can find a satisfying assignment mapping them to the same domain value that does not satisfy $(V,C)$. Finally, note that $|C| = \frac{n(n-1)}{2}$ and thus this proves $\nrd_{\{R\}}(n) \geq \frac{n(n-1)}{2}$, completing the proof.

Next, consider the relations $R$ for which we claim $\nrd_{\{R\}}(n) = \lceil\frac{n}{2}\rceil \cdot \lfloor\frac{n}{2}\rfloor$, namely $R \in \{p,\subset, s,f,d,PP,$ $NTPP\}$ and the corresponding inverse relations. For the upper bound, consider first the case where $n$ is even and consider the instance $(V,C)$ with $V=\{x_i,y_i\mid i \in [\frac{n}{2}]\}$ and $C=\{R(x_i,y_j)\mid i,j \in [\frac{n}{2}]\}$. By counting, $(V,C)$ is the only instance corresponding to a relational network $(V,E_C)$ that fits $\frac{n^2}{4}$ directed edges $(x,y,R)$ corresponding to constraints $R(x,y)$ into a graph without inducing either:
\begin{itemize}
    \item a cycle $(x_1,\ldots,x_k)$ with $R(x_i,x_{i+1 \mod k}) \in C$ for all $i \in [k]$ which would make the instance unsatisfiable and all constraints outside of this cycle are thus not non-redundant,
    \item a path $(x,x_1,\ldots,x_k,y)$ with $R(x,x_1),R(x_i,x_{i+1}),R(x_k,y),R(x,y) \in C$ for all $i \in [k-1]$ where $R(x,y)$ would not be non-redundant.
\end{itemize} 
However, by the same reasoning, there cannot exists a non-redundant instance with more than $\frac{n^2}{4}$ constraints, proving the claim. Let now $n \in \mathbb{N}$ be odd. We construct a similar instance $(V,C)$ with $V=\{x_i,y_j\mid i \in [\lceil\frac{n}{2}\rceil], j \in [\lfloor \frac{n}{2} \rfloor]\}$ and $C=\{R(x_i,y_j)\mid i \in [\lceil\frac{n}{2}\rceil], j \in [\lfloor \frac{n}{2} \rfloor]\}$. By analogous logic, we have a non-redundancy of at most $\lceil\frac{n}{2}\rceil \cdot \lfloor\frac{n}{2}\rfloor$ in this case. \\
For the lower bound, consider again the instances discussed above for even and odd $n$. Both are also non-redundant instances for these relations, showing the respective lower bound. Showing this for even and odd $n$ for all relations $R \in \{p,\subset, s,f,d,PP,NTPP\}$ would impact readability, and as the proofs are analogous, we show this as an example for even $n$ and $R = p$. Let thus $(V,C)$ be the instance from above for even $n$, i.e. with $V=\{x_i,y_i\mid i \in [\frac{n}{2}]\}$ and $C=\{R(x_i,y_j)\mid i,j \in [\frac{n}{2}]\}$. Consider an arbitrary constraint $R(x_k,y_\ell) \in C$ and the assignment $\sigma$ with
\begin{equation*}
    \sigma'(z) = \begin{cases}
        (0,1), & \text{if $z=x_i$ and $i \neq k$,}\\
        (0,2), & \text{if $z=x_i$ and $i = k$,}\\
        (3,4), & \text{if $z=y_i$ and $i \neq \ell$,}\\
        (2,4), & \text{if $z=y_i$ and $i = \ell$.}\\
    \end{cases}
\end{equation*}
For all $(x_i,y_j)$ with $i \neq k$ or $j \neq \ell$, we have $(\sigma(x_i),\sigma(y_j)) \in p$. Only if $i = k$ and $j = \ell$ do we have $(\sigma(x_i),\sigma(y_j)) \notin p$. Thus $\sigma$ is a satisfying assignment for $(V,C\setminus\{R(x_k,y_\ell)\})$ but not for $(V,C)$, proving that the constraint $R(x_k,y_\ell)$ is non-redundant. As this was an arbitrarily chosen constraint, all constraints in $C$ are non-redundant and $(V,C)$ is non-redundant instance witnessing $\nrd_{\{p\}}(n) \geq |C| = \lceil\frac{n}{2}\rceil \cdot \lfloor\frac{n}{2}\rfloor$.

For the relations $o,\alpha,TPP$, while we do not get the upper bound of $\lceil\frac{n}{2}\rceil \cdot \lfloor\frac{n}{2}\rfloor$ using the above construction, the above instance also provides the same lower bound of $\nrd_{\{R\}}(n) \geq \lceil\frac{n}{2}\rceil \cdot \lfloor\frac{n}{2}\rfloor$ by analogous argument.

For $m$, consider first the instance $(V,C)$ with $V = \{s,t,x_i \mid i \in [n-2]\}$ and $C=\{m(s,x_i), m(x_i,t) \mid i \in [n-2]\}$. Consider any interval $x_i = [x_i^-, x_i^+]$, as well as $s = [s^-,s^+]$ and $t = [t^-,t^+]$. Setting $x^- \neq s^+$ in an assignment satisfies $(V,{C\setminus\{m(s,x_i)\}})$ but not $(V,C)$. Similarly, setting $x^+ \neq t^-$ in an assignment satisfies $(V,{C\setminus\{m(x_i,t)\}})$ but not $(V,C)$. Thus $(V,C)$ is a non-redundant instance of $\csp(\{m\})$ and we have $\nrd_{\{m\}}(n) \geq |C| = 2n-4$. For the upper bound, note first that each constraint  $m([x^-,x^+],[y^-,y^+])$ enforces an equivalence upon the endpoints of the two intervals, namely $x^+ = y^-$. Let $(V,C)$ be a non-redundant instance of $\csp(\{m\})$ and w.l.o.g. let $V=\{x_1,\ldots,x_n\}$ with $x_i = [x_i^-, x_i^+]$ for all $i \in [n]$. Let furthermore $(x_i^-)^C_=$ be the equivalence class of $x_i^-$ under $C$ and respectively for $x_i^+$ for all $i \in [n]$. Any assignment that fulfills these equivalence classes is a satisfying assignment for $(V,C)$ and thus, $c = m([x^-,x^+],[y^-,y^+]) \in C$ is satisfiable if $(x_i^+)^C_= \neq (x_i^+)^{C \setminus \{c\}}_=$ and $(y_i^-)^C_= \neq (y_i^-)^{C \setminus \{c\}}_=$. Likewise, we also have $(y_i^-)^C = (x_i^+)^C_= = (x_i^+)^{C \setminus \{c\}}_= \cup (y_i^-)^{C \setminus \{c\}}$ if $c = m([x^-,x^+],[y^-,y^+]) \in C$. Note that we have $x_i^- < x_i^+$ for all $[x_i^-,x_i^+] \in V$ and we must have $(x_i^-)_=^C \neq (x_i^+)_=^C$ because of this in any satisfiable instance. Similarly, an instance is unsatisfiable exactly if there is a cycle $(x_1,\ldots,x_k)$ such that $m(x_i,x_{i+1 \mod k}) \in C$ for all $i \in [k]$. In this case, all constraints outside of such a cycle would not be non-redundant. As the number of non-redundant constraints in instances with such cycles are bounded by $n-1$, we thus want to avoid them to maximize the number of non-redundant constraints. For $n \geq 2$, there must then be at least three different equivalences under $C$ if an instance does not contain such a cycle. Finally, note that the only such instance where there are three different equivalence classes under $C$ for $n\geq 2$ is $(V,C)$ with $V=\{x,y\}$ and $C=\{m(x,y)\}$. Thus, for $n \geq 3$, we must have at least $4$ different equivalence classes of endpoints under $C$. Clearly, we then can have at most $2n-4$ non-redundant constraints, as each non-redundant constraint merges exactly two equivalence classes, which concludes the proof.

For $R\in \{EQ,e\}$, note that they enforce an equivalence between variables, i.e. if $C(x,y) = \{EQ\}$ is a constraint, then for the respective equivalence classes of $x,y$, we have $(x)_=^C = (y)_=^C$. Analogous to the above argument, we can thus prove an upper bound of $\nrd_{\{R\}}(n) \leq n-1$. The lower bound follows by Lemma~\ref{lem:redundancygeneralbounds}.
   \end{proof}

\section{Towards Single Exponential Time}

We have proven that branching, in the worst-case, is unlikely to improve upon $2^{\bigoh(n^2)}$ for  \rcci{i} ($i 
\in \{5,8\}$), $\mathfrak{A}_3$, and the full algebra $\mathfrak{A}_{13}$. This does not rule out other algorithmic ideas, however, and we now use DP to obtain $2^{\bigoh(n)}$ for $\csp(\mathfrak{A}_3^U)$ (Section~\ref{sec:ia}) and $o(n)^n$ for $\csp(\rcci{8}^U)$ (Section~\ref{sec:rcc}). We remark that the former is the first {\em unconditional} (cf. ~\cite{lagerkvist2022d}) single exponential algorithm for an NP-hard spatiotemporal reasoning problem
and the latter the first major improvement for \rcci{8}.

\subsection{$\mathfrak{A}_3$ and the Graph Interval Sandwich Problem} 
\label{sec:ia}
\begin{algorithm}[tb]
    \caption{Solving $\csp(\mathfrak{A}^U_3)$ in $\bigoh^*(4^n)$ time.}
    \label{alg:partCSPA3}
    \flushleft\hspace*{\algorithmicindent}\textbf{Input:} Variables $V$, Network $C:V^2 \rightarrow \mathfrak{A}^U_3$ \\
    \hspace*{\algorithmicindent}\textbf{Output:} $True$ exactly if the network is satisfiable
    \begin{algorithmic}[1]
    \State $V_e \leftarrow \{x^- \mid x \in V\} \cup \{x^+ \mid x \in V\}$
    \State $S \leftarrow \operatorname{Queue}((\emptyset, V_e))$
    \State $M \leftarrow \{(\emptyset, V_e)\}$
    \While{$S \neq \emptyset$ } \label{line:3:alg3partCSP}
    \State $(V_<,V_>) \leftarrow S.\operatorname{pop}()$
    \If{$V_< = V_e$} 
    \State \Return $True$
    \EndIf
    \ForAll{$x \in V_>$}
    \State $s' \leftarrow (V_< \cup \{x\}, V_> \setminus \{x\})$
    \If{$(V_e,<_{s'})$ is consistent and $s' \notin M$} \label{line:7:alg3partCSP}
    \State $S.\operatorname{append}(s')$
    \State $M \leftarrow M \cup \{s'\}$
    \EndIf
    \EndFor
    \EndWhile
    \State \Return $False$
    \end{algorithmic}
\end{algorithm}

It is known that every solution to an \textsc{IA} instance $(V,C)$ corresponds to an ordering of the $2n$ interval endpoints $x^-,x^+$ of intervals $x=(x^-,x^+) \in V$ on the rational line. In the case of $\mathfrak{A}_3$, all endpoints can moreover w.l.o.g. be assumed to be distinct  and every relation between intervals can be uniquely determined by only considering relations between three of the four endpoints, albeit this does not hold for an arbitrary selection of three endpoints. 
E.g. if we have $y^-<y^+<x^+$, we might either have $y \cap x$ or $y ~p ~x$, but if we have $y^-<x^-<y^+$, the relation of the intervals is $y \cap x$. 
Here, we construct all endpoint orderings in a DP fashion by moving endpoints from the right to the left one endpoint at a time, i.e. constructing the ordering linearly, appending one endpoint at a time to it. To ensure that each ordering still corresponds to a solution, we fix the relations for each endpoint $x$ moved by considering all triples of endpoints $(x,y^-,y^+)$, for all other intervals $(y^-,y^+)$. Still, constructing all possible endpoints orderings like this would take too long, as there are $2|V|!$ different orderings in the worst case. To work around this issue, we only consider one total order for each subset of $V$. When constructing the total order of endpoints, we partition $V$ into a set $V_<$ of endpoints that are already totally ordered by the so far constructed order and another set $V_>$ that is not yet ordered. E.g. if we have so far constructed the total order $x_1 < x_2 < x_3$, then $V_< = \{x_1,x_2,x_3\}$. For each 2-partitions
$(V_<,V_>)$, we then only save one total ordering giving us a runtime of $\bigoh^*(4^n)$. We show later that if there exists a total order on $V$, this approach suffices to find it. If there exists no total order corresponding to a satisfying instance, the instance must be unsatisfiable. Then, we will not find such a total order under this restriction either and rightfully reject the instance. Note also that for all $x \in V_<$, we have $x < y$ for all $y \in V_>$ in the final total order of all endpoints that extends the one corresponding to $(V_<,V_>)$.

Formally, let $(V,C)$ be a $\csp(\mathfrak{A}^U_3)$ instance and $V^- = \{x^- \mid x \in V\}, V^+=\{x^+ \mid x \in V\}$ the sets of endpoints. 
We relate to each 2-partition $s = (V_<,V_>)$ of $V^- \cup V^+$ created by moving endpoints $x_1,\ldots,x_\ell$ from $V^- \cup V^+$ to $\emptyset$ sequentially the partial order $(V^- \cup V^+, <_s)$ corresponding to the sequence $x_1,\ldots,x_\ell$ of endpoints used to construct it. As we consider one such sequence for each 2-partition $s = (V_<,V_>)$, $<_s$ is well-defined. We call $s$ {\em consistent} exactly if for all variables $x,y \in V$ and the relations $R_e(x,y)$ imposed by $(V^- \cup V^+, <_s)$, we have $R_e(x,y) \in \{R_1, \ldots, R_m\} \cup \{\bot\}$ for each $R^U_{\{R_1, \ldots, R_m\}}(x,y) \in C$, where $R_e(x,y) = \bot$ if the order of endpoints of $x,y$ in $(V^- \cup V^+, <_s)$ does not induce a relation in $\mathfrak{A}_3$.

\begin{theorem} \label{theo:A3singleexp}
    Algorithm~\ref{alg:partCSPA3} solves $\csp(\mathfrak{A}_3)$ in time $\bigoh^*(4^n)$.
\end{theorem}
\begin{proof} 
    We already argued soundness above. It remains to prove that Algorithm~\ref{alg:partCSPA3} always finds a satisfying total order if it exists. Consider two different consistent partial orders $(V^- \cup V^+, <_s),(V^- \cup V^+, <_{s'})$ corresponding to $s=s'=(V_<,V_>)$ such that $(V^- \cup V^+, <_s)$ can be expanded to total order corresponding to a solution, while $(V^- \cup V^+, <_{s'})$ cannot. All relations including intervals $x \in V$ with $x^-,x^+ \in V_<$ are decided (i.e. fixed by $<_s,<_{s'}$) and as the orders are consistent, these cannot cause any inconsistencies in the future. Thus there must be some \textit{open} intervals $x,y \in V$ with $x^-,y^- \in V_<$ but $x^+,y^+ \in V_>$ such that $x^- <_s y^-$ but $y^- <_{s'} x^-$. But in either case, we already fix $x \cap y$ under $<_s,<_{s'}$ and $x,y$ have the same relations with all $z \in V$ with $z^-,z^+ \in V_>$ in the same expansions (i.e. if we append the same remaining endpoint orderings to $<_s,<_{s'}$). So $(V^- \cup V^+, <_s),(V^- \cup V^+, <_{s'})$ must either both be expandable to a satisfying total order or not. We conclude that if $(V,C)$ is satisfiable, at least one sequence $x_1,\ldots,x_{2n} \in V^- \cup V^+$ corresponding to a satisfying assignment for $(V,C)$ is \textit{found} (as a series of 2-partitions) by Algorithm~\ref{alg:partCSPA3}. 
    Finally, as the set $S$ contains at most every two partition of $V^- \cup V^+$, we have a bound of at most $2^{|V^- \cup V^+|} = 4^n$ iterations of the outer loop in line~\ref{line:3:alg3partCSP}.

    \begin{algorithm}[tb]
    \caption{Compute an instantiation $\sigma_s$ for $s = (V_<,V_>)$. Symmetric assignments are omitted.}
    \label{alg:partialinstantA3}
    \flushleft\hspace*{\algorithmicindent} \textbf{Input: A two-partition $(V_{<}, V_{>})$, A consistent instantiation $\sigma_{s'}$ for $s' = (V_{<} \setminus\{x\}, V_> \cup \{x\})$}
    \begin{algorithmic}[1]
    \State $\sigma_{s}: V^2 \rightarrow \{p,\cap,p^{-1}\}, (x,y) \mapsto \bot$
    \ForAll{$x,y \in V$ with $\sigma_{s'}(x,y) \neq \bot$}
    \State $\sigma_s(x,y) \leftarrow \sigma_{s'}(x,y)$
    \EndFor
    \If{$x$ is an interval start point}
    \ForAll{$y^+ \in (V_{<}\setminus\{x\})$}
    \State $\sigma_i(y,x) \leftarrow ~p$
    \EndFor
    \Else
    \ForAll{$y^- \in V_{>}$}
    \State $\sigma_i(x,y) \leftarrow ~p$
    \EndFor
    \EndIf
    \ForAll{$y^- \in V_{<}\setminus\{x\}$ with $y^+ \in V_{>}$}
    \State $\sigma_i(x,y) \leftarrow ~\cap$
    \EndFor
    \State \Return $\sigma_s$
    \end{algorithmic}
\end{algorithm}

It remains to proof that it is possible to compute whether $(V^- \cup V^+, <_s)$ is consistent in polynomial time. We do this by presenting an algorithm computing all relations $R_e$ fixed by a partial order of endpoints $(V^- \cup V^+, <_s)$ in polynomial time. Let $(V^- \cup V^+, <_s)$ be a partial order corresponding to a partition $s=(V_<,V_>)$ and let $(V_<,V_>)$ be formed from $s'=(V_<\setminus\{x\},V_>\cup\{x\})$ by moving $x$ from the right side to the left side. By induction, we have already shown that the partial order $(V^- \cup V^+, <_{s'})$ corresponding to $s'$ is consistent previously, which is true if $s' = (\emptyset, V^- \cup V^+)$ as we are dealing with a satisfiable instance by assumption.\\
We correspond to each partial order $(V^- \cup V^+, <_{t})$ a partial instantiation function $\sigma_t: V^2 \rightarrow \{p,\cap,p^{-1}\}$ corresponding to the endpoint relations $R_e$ fixed by that partial order already and call $\sigma_t$ consistent exactly if $(V^- \cup V^+, <_{t})$ is consistent. It can be easily seen that $\sigma_t$ can be computed by Algorithm~1 in polynomial time, i.e. the endpoint relations enforced by $(V^- \cup V^+, <_{t})$ are exactly those on which $\sigma_t$ as computed by Algorithm~1 is defined on.
To now compute whether $(V^- \cup V^+, <_s)$ is consistent, we first check whether $x$ is an interval endpoint $x^+$ and $V_<$ does not contain its respective startpoint $x^-$ yet and reject if this happens. In other cases, we compute $\sigma_s$ based on $\sigma_{s'}$ using Algorithm~\ref{alg:partialinstantA3}, rejecting if Algorithm~\ref{alg:partialinstantA3} overwrites any value of $\sigma_{s'}$, i.e. if $\sigma_s$ does not extend $\sigma_{s'}$. Finally, we check whether $\sigma_s$ is consistent with $(V,C)$, i.e. we also have $\sigma(x,y) \in \{R_1,\ldots,R_m\}$ for all constraints $R^U_{R_1,\ldots,R_M} \in C$ for all $x,y \in V$ where $\sigma_s$ is defined. If this is true, we accept and conclude that $(V^- \cup V^+, <_s)$ is consistent, otherwise we return that it is not. The correctness of this approach is obvious and follows from the above argument. Finally, Algorithm~\ref{alg:partialinstantA3} clearly runs in polynomial time.
\end{proof}

We can  extend this as follows. Call a graph $G=(V,E)$ an {\em interval graph} if and only if $(V,C_E)$ is a satisfiable $\csp(\mathfrak{A}_3^U)$ instance, where $\cap(v,w) \in C_E(v,w)$ if  $(v,w)\in E$ and 
            $R_{\{p,p^{-1}\}}^U(v,w) \in C_E$ otherwise.
In the {\em interval graph sandwich problem} (IGSP) we are then given
two graphs $G_1=(V,E_1)$ and $G_2=(V,E_2)$ with $E_1 \cap E_2 = \emptyset$, and want to decide if there is
an interval graph $G=(V,E)$ with $E_1 \subseteq E \subseteq E_1 \cup E_2$.
IGSP is known to be a special case of $\csp(\mathfrak{A}^U_3)$~\cite{IGSPpaper}, which, together with Theorem~\ref{theo:A3singleexp}, yields the following corollary.

\begin{col}
    $\textsc{IGSP}$ is solvable in $\bigoh^*(4^{n})$ time where $n = |V|$ is the number of vertices.
\end{col}

\subsection{An Improved Algorithm for RCC-8}
\label{sec:rcc}
As the final result we present an improved DP algorithm for $\csp(\rcci{8}^U)$.
We recall that $\csp(\rcci{8}^U)$ can be solved in $\bigoh((0.531n)^n)$ time by enumerating so-called {\em ordered partitions}, or total orders~\cite{IncompRelation}. Here, the intuition is that  for a given total order $(V,<_T)$, if $x<_Ty$, then the relation between $x$ and $y$ in our 
instance is neither $EQ$, $TPP^{-1}$, nor $NTPP^{-1}$, and the resulting instance can be solved in polynomial time. I.e., total orders function as \emph{certificate}s. 
We thus need to significantly decrease how many total orders we need to handle.
We work bottom-up towards defining a recurrence relation $R_f(S)$, by focusing on how we compare different total orders to each other.
However, if one attempts to apply standard DP with total orders as states we quickly run into problems. 
Normally there are methods to quickly choose between two possible states, but this is not the case here since choices we make early can have a significant impact later.

\begin{example}\label{ex:RCCex}

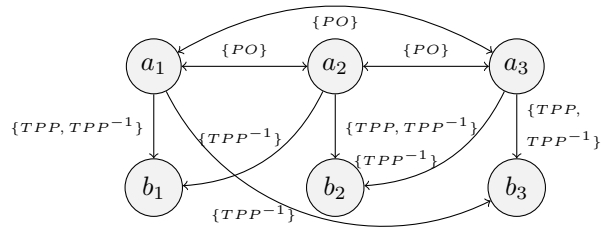
\begin{figure}[b!]
\centering
  \begin{tikzpicture}[scale=0.8]

    \node[circle,draw, fill = gray!10] (a1) at (0,0) {$a_1$};
    \node[circle,draw, fill = gray!10] (a2) at (3,0) {$a_2$};
    \node[circle,draw, fill = gray!10] (a3) at (6,0) {$a_3$};
    
    \node[circle,draw, fill = gray!10] (b1) at (0,-2) {$b_1$};
    \node[circle,draw, fill = gray!10] (b2) at (3,-2) {$b_2$};
    \node[circle,draw, fill = gray!10] (b3) at (6,-2) {$b_3$};
    
    \draw[->] (a1) -- (b1) node[midway, left] {\tiny $\{TPP,TPP^{-1}\}$};
    \draw[->] (a2) -- (b2) node[midway, right] {\tiny $\{TPP,TPP^{-1}\}$};
    \draw[->] (a3) -- (b3) node[midway, right] [align=left]{\tiny $\{TPP,$\\\tiny $TPP^{-1}\}$};

    \draw[->] (a2)  edge[bend right = -30] node [left, pos=0.3] {\tiny $\{TPP^{-1}\}$} (b1);
    \draw[->] (a3)  edge[bend right = -30] node [left] {\tiny $\{TPP^{-1}\}$} (b2);

    \draw[<->] (a1) -- (a2) node[midway, above] {\tiny $\{PO\}$};
    \draw[<->] (a2) -- (a3) node[midway, above] {\tiny $\{PO\}$};
    \path[<->] (a1) edge[bend left] node [below] {\tiny $\{PO\}$} (a3);

    \path[->] (a1) edge[bend right = 45] node [left] {\tiny $\{TPP^{-1}\}$} (b3);

    \end{tikzpicture}
  \caption{Graphic representation of the instance described in Example~\ref{ex:RCCex}.}
  \label{fig:RCCex}
\end{figure}

Consider a $\csp(\rcci{8}^U)$ instance $(\{a_1,a_2,a_3,b_1,b_2,b_3\},C)$ with relational network $c$ given by $c(a_1,b_1)=\{TPP,TPP^{-1}\}\in C$, $c(a_2,b_2)=\{TPP,TPP^{-1}\}\in C$, $c(a_3,b_3)=\{TPP,TPP^{-1}\}\in C$, $c(a_1,b_3)=\{TPP^{-1}\}\in C$, $c(a_2,b_1)=\{TPP^{-1}\}\in C$, $c(a_3,b_2)=\{TPP^{-1}\}\in C$, $c(a_1,a_2)=\{PO\}\in C$, $c(a_2,a_3)=\{PO\}\in C$, and $c(a_1,a_3)=\{PO\}\in C$, as shown in Figure~\ref{fig:RCCex}.
Assume that we try to handle this with standard DP.
Also, assume we have placed variables $b_1$, $b_2$ and $b_3$, and now are placing $a_1$ and are hence assuming $b_1$ is $TPP$ of $a_1$.
Similarly, there is another case where we have already placed $a_1$, $b_2$ and $b_3$ and are now placing $b_1$ and are hence assuming $a_1$ is $TPP$ of $b_1$.
Both of these cases are valid, and in the standard case they both lead to the same state.
But assuming $a_1$ is $TPP$ of $b_1$ leads to an unsolvable case, while assuming $b_1$ is $TPP$ of $a_1$ does not.
Hence, if we choose to only remember the incorrect case of $a_1$ (is $TPP$ of $b_1$) we cannot guarantee correctness.
The obvious solution to the problem is to remember all partial solutions, but this results in an algorithm with a run time of the form $\bigoh(n)^n$, which does not improve the existing upper bound.
\end{example}

Hence, for each set $S\subseteq V$, $R_f(S)$ thus has to represent a set of total orders.
Our goal is hence to minimize the size of these sets without losing correctness.
We resolve this by comparing the total orders using \emph{what is not yet solved} which then allows us to only keep a subset of all possible total orders.

As a road map, we first define how an instance behaves under an assumed total order followed by
showing that a sub-instance can be solved in polynomial time and that the solution in some sense is as general as possible.
We then introduce the crucial concept of {\em inconsistency paths} (IPs):
a potential path of transitivity beginning in the part of the instance that has not yet been considered, but such that the path ends in the part that has already been handled.
Using IPs, we then see how we can compare total orders
and more importantly, how we can minimize how many of these total orders $R_f(S)$ needs to contain.
This works since if we have two total orders $T=(V_T,\leq_T)$ and $T'=(V_T,\leq_T')$ and we know that all IPs relevant for $T$ also exist in $T'$, then $T$ is at least as good as $T'$. 
Hence, by minimizing $R_f(S)$ according to this, we can guarantee that we keep enough to ensure that we can construct a solution to the problem (if one exists).
Last, we analyze the sizes of these minimal $R_f(S)$ in order to  achieve the desired complexity. 

We start by defining how  we apply a partial order to a given relational network.
\begin{definition} \label{def:FunderP}
 Given a $\csp(\rcci{8}^U)$ instance $I=(V,C)$ with multi relational network $f$ and a partial order $P=(V,\leq_P)$, we for all $x,y \in V$ define the instance 
 ${(P\circ f)(x,y) =}$
 \begin{equation*}
    \begin{cases}
        f(x,y)\setminus\{TPP^{-1},NTPP^{-1},EQ\}, &\text{ if } x<_Py, \\
        f(x,y)\setminus\{TPP,NTPP,EQ\}, &\text{ if } y<_Px, \\
        f(x,y)\cap\{EQ\}, &\text{ if } x=_Py, \\
        f(x,y), &\text{otherwise.}
    \end{cases}
\end{equation*}

Additionally, given two relational networks $f$ and $g$, let $(f\cap g)(x,y) = f(x,y)\cap g(x,y)$.
\end{definition}

Given this we can now show that if an arbitrary instance $I$ is a yes-instance (i.e. there exists an assignment satisfying $I$), there will be a total order $T$, such that $I$ under $T$ is also a yes-instance.

\begin{lem}\label{lem:IyesTyes}
If a $\csp(\rcci{8}^U)$ instance $I$ with multi relational network $f$ is a yes-instance then there is a total order $T$ such that $T\circ f$ is a yes-instance.
\end{lem}
\begin{proof} (Sketch)
Any relational network satisfying $I$ describes a partial order over $V$. Topologically sorting this partial order yields the desired properties.
More precisely, take a complete relational network $g$ over $V$, satisfying $I$, i.e., $|g(x,y)|=1$, $g(x,y) \in c(x,y)\in C$ for all $x,y\in V$ and $g$ being transitive.
Now any $TPP$, $NTPP$ and $EQ$ relations in $g$ describe a partial ordering $P=(V_P,<_P)$ over $V$.
Topologically sort $P$ to a total order $T$ such that if $x=_Py$ then $x=_Ty$.
Since $x\not <_Ty$ if $g(x,y)\in\{TPP,NTPP\}$, $g$ must be consistent with $T\circ f$ and hence $g$ also satisfies $T\circ f$, which proves the lemma.
\end{proof}

If $P=(V_P,\leq_P)$ is a partial order then we write $P<S$ for the partial order $(V_P\cup S,\leq_{P'})$ where $\leq_{P'} = \leq_P \cup \{(x,y)\,|\,x\in V_P \text{ and }y\in S\}$.
Hence, we may write $(P<S)\circ f$ if $V_P\cup S = V$.

\begin{definition} \label{def:FunderT}
 Given a $\csp(\rcci{8}^U)$ instance $I=(V,C)$ with multi relational network $f$ and a total order $T=(V_T,\leq_T)$ with $V_T\subseteq V$, we for all $x,y \in V$ define 
  ${((T<V\setminus V_T)\odot f)(x,y) =}$
 \begin{equation*}
    \begin{cases}
        f(x,y)\setminus\{TPP^{-1},NTPP^{-1},EQ\}, &\text{ if } x<_Ty, \\
        f(x,y)\setminus\{TPP,NTPP,EQ\}, &\text{ if } y<_Tx, \\
        f(x,y)\cap\{EQ\}, &\text{ if } x=_Ty, \\
        \rcci{8}, &\text{otherwise.}
    \end{cases}
\end{equation*}
Note that we assume that $x<_Ty$ for all $x\in V_T$ and $y \in V\setminus V_T$ to simplify the notation.

Additionally we say that $(T<V\setminus V_T)\odot f$ is \emph{reduced locally $k$-consistent} if for every set $S\subseteq V$ with $|S|\leq k$ then for every pair $x,y\in S$ and for every relation $r\in ((T<V\setminus V_T)\odot f)(x,y)$ then the instance $g$ over $S$ defined as $g(u,v)=((T<V\setminus V_T)\odot f)(u,v)$ for all $u,v\in S$, $u,v\neq x,y$ and $g(u,v) = r$, is a yes-instance.
Any relation $r\in ((T<V\setminus V_T)\odot f)(x,y)$ not satisfying the above criteria is \emph{$k$-excessive}. 
\end{definition}
While the two definitions of $(P<V\setminus V_P)\circ f$ and $(T<V\setminus V_T)\odot f$ may seem very similar, there are some key differences between them:
$(P<V\setminus V_P)\circ f$ can be hard to solve since if $P=(\emptyset,\emptyset)$ then $(P<V\setminus V_T)\circ f=f$, and $(P<V\setminus V_P)\circ f \Rightarrow f$. In contrast, as we will soon prove, $(T<V\setminus V_T)\odot f$ is tractable, but $(T<V\setminus V_T)\odot f \not \Rightarrow f$. 
Additionally, note that if $V_T=V$ then $T\odot f = T\circ f$. 
This means that $(T<V\setminus V_T)\odot f$ will be easy to work with and preferable to use, but we will not be sure we are actually on the correct track until $V_T=V$.
Given an arbitrary $\csp(\rcci{8}^U)$ instance $I=(V,C)$ with multi relational network $f$ and a total order $T=(V_T,<_T)$, let $\mathit{RLC_k}((T<V\setminus V_T)\odot f)$ be a function outputting $(T<V\setminus V_T)\odot f$ with all $k$-excessive relations removed.

\begin{lem} \label{lem:local_consistency}
The function $\mathit{RLC_k}((T<V\setminus V_T)\odot f)$ can be calculated in $\bigoh(poly(|V|))$ time for any constant $k$.
\end{lem}
\begin{proof} 
Let $g=(T<V\setminus V_T)\odot f$. We claim that for any subset of $V$ of size $k$ any $k$-excessive relations can be removed in polynomial time by brute-force checking. As there are $|V|^k$ of these subsets it is all in $\bigoh(poly(|V|))$ which completes the proof.
To show this, we enumerate all $S\in V$ with $|S|\leq k$, and for every $x,y\in S$, we force a relation on $r\in g(x,y)$ and solve the remaining instance over $S$. If the instance is solvable, keep $r\in g(x,y)$, otherwise update $g(x,y)$ to $g(x,y)\setminus r$. This is continued until no more changes occur.
Enumerating $S$ can clearly be done in polynomial time since $k$ is constant.
Similarly, any $\csp(\rcci{8}^U)$ of size at most $k$ can clearly be solved in constant time.
Each iteration over all $S$ either decreased $|g(x,y)|$ by at least one for some $x,y\in V$, otherwise no change occurs and we stop.
Hence, the maximum number of iterations is bounded by a polynomial function of $|V|$.
\end{proof}

We order relations in the following way: $DC=\diamond_1$, $EC=\diamond_2$, $PO=\diamond_3$, $TPP=\diamond_4$, and $NTPP=\diamond_5$. Given  $\diamond \subseteq \{\diamond_1,\dots,\diamond_5\}$  we let $\diamond_{min}\in \diamond$ be the smallest element under this ordering.

 \begin{obs}\label{obs:simplercc8transiticty}
By taking the sequence $DC=\diamond_1$, $EC=\diamond_2$, $PO=\diamond_3$, $TPP=\diamond_4$, and $NTPP=\diamond_5$, and studying the transitivity of the relations of \rcci{8} under some total order $T$ we can observe the following important properties:
\begin{itemize}
    \item For every combination of $i\in[2,\dots,5]$ and $j\in[5]$, the set of allowed relations for $x\diamond z$ given $x\diamond_iy$ and $y\diamond_jz$ is a subset of the allowed relations for $x\diamond_{i-1}y$ and $y\diamond_jz$.
    \item For $i\in[5]$ and every pair $j,j'\in[5]$ with $j<j'$, let $\diamond_{i,j}$ and $\diamond_{i,j'}$ be the sets of relations allowed between $x$ and $z$ given $x\diamond_iy$ and $y\diamond_jz$, and given $x\diamond_iy$ and $y\diamond_{j'}z$ respectively. For every $\diamond_k \in \diamond_{i,j'}$ then either $\diamond_k \in \diamond_{i,j}$, or for all $\diamond_{k'} \in \diamond_{i,j}$  we have $k'<k$. 
\end{itemize}
 \end{obs}

 We now establish correctness of Algorithm~\ref{alg:polyRN} which shows that  $(T<V\setminus V_T)\odot f$ can be solved in polynomial time.

\begin{algorithm}[tb]
    \caption{Solving $(T<V\setminus V_T)\odot f$ in $\bigoh(poly(||I||))$}
    \label{alg:polyRN}
    \flushleft\hspace*{\algorithmicindent} \textbf{Input: Multi relational network $g=(T<V\setminus V_T)\odot f$}
    \begin{algorithmic}[1]
    \State $g' \leftarrow RLC_3(g)$
    \For{$x_i \in x_1<_T\ldots<_Tx_{|V_T|}$}
    \For{$x_j \in x_{i-1}>_T\ldots>_Tx_{1}$}
    \State $g'(x_j,x_i) \leftarrow \{g'(x_j,x_i)_{min}\}$\label{line:polyRNAssign}
    \State $g' \leftarrow RLC_3(g')$
    \EndFor
    \EndFor
    \State \Return $g'$
    \end{algorithmic}
\end{algorithm}

\begin{lem}\label{lem:todotfpoly}
Given a $\csp(\rcci{8}^U)$ network $(T<V\setminus V_T)\odot f$ for some total order $T=(V_T,\leq_T)$ and network $f$, Algorithm~\ref{alg:polyRN} returns $g(x,y)=\emptyset$ for all $x,v\in V_T$ if $(T<V\setminus V_T)\odot f$ is a no-instance.
Otherwise Algorithm~\ref{alg:polyRN} returns a relational network $g$ such that 
\begin{itemize}
    \item $g$ satisfies $(T<V\setminus V_T)\odot f$, and
    \item there is no relational network $g'\not = g$ satisfying $(T<V\setminus V_T)\odot f$ such that there is a pair $x,y\in V$ where if $g'(x,y)=\{\diamond_i\}$ and $g(x,y)=\{\diamond_j\}$ then $i<j$. 
\end{itemize}
\end{lem}
\begin{proof}
$\mathit{RLC_k}$ is correct by definition as it only removes relations that cannot be used to satisfy the original instance.
Additionally, if the relational network ever becomes inconsistent, every relation is $\emptyset$.
For the remainder of the algorithm it comes down to three things: assigning $g(x_i,x_j)$, $g(x_j,x_k)$ and $g(x_i,x_k)$ for $i<j<k$.
By Line~\ref{line:polyRNAssign} and the ordering of the looping, we ensure $g(x_i,x_j)$ is always assigned first, and assigned according to the first of $DC$, $EC$, $PO$, $TPP$, and $NTPP$, as desired by our stated properties.
By Observation~\ref{obs:simplercc8transiticty}, we can see that this ensures maximum allowed relations for $g(x_i,x_k)$, as well as allowing $g(x_j,x_k)$ to be as early relations from $DC$, $EC$, $PO$, $TPP$, and $NTPP$ as possible. 
Hence, assigning relations in this way fulfills the properties stated in the lemma.
\end{proof}

Now we are finally ready to introduce inconsistency paths, i.e., our representation of what is not yet solved. 

\begin{definition}\label{def:incPaths}
Given a $\csp(\rcci{8})$ network $f$ over $V$ and a partial order $P$ over $V$ such that $x_1<_{P}\dots<_{P}x_i$, let $R_t(x_1,\ldots,x_i)$ be the set of relations allowed between $x_1$ and $x_i$ given transitivity when going from $x_{j}$ to $x_{j-1}$ for all $j\in[2,\dots,i]$ and the assumption that $P$ is applied to $f$, i.e. that the network is now $P\circ f$.
    An \emph{inconsistency path} (IP) for a relational network $g$ satisfying $(T<V\setminus V_T)\odot f$ for some relational network $f$ and total order $T=(V_T,\leq_T)$, is a pair $(h,t)$ such that there is a sequence $x_1<_T\dots<_Tx_k< t$, $\{x_1,\dots x_k\}\subseteq V_T$, and a sequence $y_1,\ldots,y_i$, $t=y_1$, $h=y_i$, $\{y_1,\dots y_i\}\subseteq V\setminus V_T$ such that 
    $g(x_1,h)\cap R_{t}(x_1,\dots x_k,t,\dots,h) = \emptyset$.
\end{definition}

Let $\mathcal{I}_f(g)$ be the set of all IPs in $g$. Finding all IPs, and hence everything that is yet to be solved, turns out to be easy.

\begin{lem}\label{lem:I_f}
    $\mathcal{I}_f(g)$ can be calculated in $\bigoh(poly(V))$ time.
\end{lem}
\begin{proof}
    There are $|V\setminus V_T|^2$ possible IPs and checking them can be done by e.g. breadth-first search.
\end{proof}

We also need to compare different relational networks.

\begin{definition}
    We say $g_1\prec_f g_2$ if $\mathcal{I}_f(g_1) \subset \mathcal{I}_f(g_2)$, $g_1 \equiv_f g_2$ if $\mathcal{I}_f(g_1) = \mathcal{I}_f(g_2)$, and $g_1 \sqcap_f g_2$ if neither $g_1\prec_f g_2$, $g_2 \prec_f g_1$ nor $g_1 \equiv_f g_2$, i.e.\ $g_1$ and $g_2$ are incomparable.
\end{definition}

    Note that $g_1$ and $g_2$ are not necessarily over the same total orders $T_1=(V_T,\leq_{T_1})$ and $T_2=(V_T,\leq_{T_2})$. Now, by Lemma~\ref{lem:todotfpoly} we know that we can get the best relational network for each $T$.

\begin{lem}\label{lem:RNoptimalFormT}
    The relational network $g$ over $T$ returned by Algorithm~\ref{alg:polyRN} is $\prec_f$-minimal, i.e. there is no other relational network $g'$ over $T$ that also satisfies $(T<V\setminus V_T)\odot f$ and such that $g'\prec_f g$.
\end{lem}
\begin{proof} 
First, note again that Algorithm~\ref{alg:polyRN} returns the relational network with every relation chosen as early as possible from the sequence $DC$, $EC$, $PO$, $TPP$, and $NTPP$ as by Lemma~\ref{lem:todotfpoly}. Let $g$ be the network returned by Algorithm~\ref{alg:polyRN} and assume that there exists $g' \prec_f g$ which also satisfies $(T < V \setminus V_t) \odot f$.
There is then a pair $x,y\in V_T$ such that $g(x,y)$ comes after $g'(x,y)$ in the sequence.
Otherwise there can be no $(h,t)\in \mathcal{I}_f(g)$ while $(h,t)\not\in \mathcal{I}_f(g')$, as noted by Observation~\ref{obs:simplercc8transiticty}.
However, as Algorithm~\ref{alg:polyRN} guarantees the first possible choice from the sequence, such a $g'$ cannot exist, giving us a contradiction.
\end{proof}

With Lemma~\ref{lem:RNoptimalFormT} we can now use $T$ in place of any relational network $g$ satisfying $(T<V\setminus V_T) \odot f$, i.e., we can assume that $\mathcal{I}_f(T)=\mathcal{I}_f(g')$, where $g'$ is the output of Algorithm~\ref{alg:polyRN}.

\begin{definition}\label{def:sqcap_f}
    Given a set of total orders $\mathbf{T}$ over a set of variables $V_T\subseteq V$ and a multi relational network $f$ over $V$, let $\sqcap_f^{min}(\mathbf{T})$ be a maximal subset of $\mathbf{T}$ such that (1) for no $T,T'\in \sqcap_f^{min}(\mathbf{T})$ we have $T \equiv_f T'$ and (2) for every $T \in \sqcap_f^{min}(\mathbf{T})$ there is no $T'\in \mathbf{T}$ such that $T' \prec_f T$.
    Additionally, if $(T<V\setminus V_T)\odot f$ is a no-instance, then $T\not \in \sqcap_f^{min}(\mathbf{T})$.
\end{definition}
Note here the last property, as this will allow us to keep only total orders that are actually of interest to us. We can achieve, and keep, this property by running Algorithm~\ref{alg:polyRN} on every total order. 
With the minimal sets defined, we also want to know what the cost is to compute them. 

\begin{lem}\label{lem:PCleanComplexity}
Given a set $\mathbf{T}$ of total orders over a set of variables $V_T\subseteq V$ and a multi relational network $f$ over $V$, we can compute $\sqcap_f^{min}(\mathbf{T})$ in $\bigoh^*(|\mathbf{T}|^2)$ time.
\end{lem}
\begin{proof}
We can brute-force compare all orders in $\mathbf{T}$.
Comparing two orders can be done in polynomial time by Lemma~\ref{lem:I_f}, and the last property of Definition~\ref{def:sqcap_f} is done by Algorithm~\ref{alg:polyRN}, which is polynomial in $|V|$.
\end{proof}

While the factor $|\mathbf{T}|^2$ does not seem too costly, we will later see that this is actually the main bottleneck of the final algorithm, and any improvement here would have a significant impact.
Before turning to the the main result of this section we introduce two simplifying notations.

\begin{definition}
    For any set $S$, let $S^=$ be the total order $(S,\{x=y\,|\,x,y\in S\})$, i.e. the total order where all variables are equivalent.
    For a total order $T=(V_T,\leq_T)$ (assuming $V_T \cap S = \emptyset$) we then by $T<S^=$ mean the total order constructed by first taking $T$, and then adding the set $S$ as a new partition that comes after all the ones in $T$. 
\end{definition}

We are now ready to define the main recurrence relation $R_f(S)$ (which constitutes our algorithm) as
\begin{equation*}
    \begin{cases} 
        \{(\emptyset,\emptyset)\} \text{ if }S=\emptyset\text{, and otherwise} \\ 
        \sqcap_f^{min}(\bigcup_{V_T\subseteq S} \{(T<(S\setminus V_T)^=)\,|\,T\in R_f(V_T)\})
    \end{cases}
\end{equation*}

First we establish correctness.

\begin{lem}\label{lem:PoTCorrect}
    If, and only if, $R_f(V)\not = \emptyset$ then $f$, and hence $I=(V,C)$, is a yes-instance.
\end{lem}
\begin{proof}
    By Lemma~\ref{lem:IyesTyes} we know that $T$ exists, and if $R_f(V)\not = \emptyset$, then, since $\sqcap_f^{min}$ calls Algorithm~\ref{alg:polyRN}, we must have a yes-instance by Lemma~\ref{lem:todotfpoly}.
    For the other direction the only thing that can introduce non-correctness is going from $\mathbf{T}$ to $\sqcap_f^{min}(\mathbf{T})$.
    Assume $(T'<V\setminus V_T)\circ f$ is a yes-instance while $(T<V\setminus V_T)\circ f$ is a no-instance, with $T \prec_f T'$, $T,T'\in\mathbf{T}$, and $T'\not \in\sqcap_f^{min}(\mathbf{T})$.
    Hence, there must be a function $h \colon (V\setminus V_T)^2 \xrightarrow{} \{DC,EC,PO,TPP,NTPP,TPP^{-1},NTPP^{-1},EQ\}$ such that $((T'<V\setminus V_T)\circ f)\cap h$ is a yes-instance.
    Since $\mathcal{I}_f(T) \subset \mathcal{I}_f(T')$ then $((T<V\setminus V_T)\circ f)\cap h$ must also be a yes-instance. 
    Hence, $(T'<V\setminus V_T)\circ f \Rightarrow (T<V\setminus V_T)\circ f$.
\end{proof}

What remains to prove is the bound on the size of $\sqcap_f^{min}(\mathbf{T})$, since this is extremely relevant for the complexity analysis of the algorithm.
\begin{lem}\label{lem:PupperBound}
Given a set $\mathbf{T}$ of total orders over a set of variables $V_T\subseteq V$ and an multi relation network $f$ over $V$ with $n=|V|$ then
$|\sqcap_f^{min}(\mathbf{T})|\in \bigoh^*((cn/\log n)^{n/2})$ for some constant $c$.
\end{lem}
\begin{proof}
Proof by contradiction:
let $g(m,k')$ be the maximum size of $\sqcap_f^{min}(\mathbf{T})$ given total orders of length $m=|V_T|$ and $k'=|V \setminus V_T|$. Hence, we now want to find the maximum for the function $g(m,k')^{(n-k')/m}$, i.e. we are assuming we are partitioning $V$ into subsets of size $m$, and assigning a total order over each of these sets, using at most $k'$ IPs for each.
To see that $k!^{(n-k(k-1))/k}$ is the upper bound, we assume $m>k$ is the smallest $m$ such that $g(m,k(k-1))>k!^{m/k}$.
Hence, $g(m-1,k(k-1))\leq k!^{(m-1)/k}$ by our assumption.
If we now try to remove an arbitrary variable $x$ from each of the total orders in $\sqcap_f^{min}(\mathbf{T})$, we will be left with a new set of total orders we call $\mathbf{T}_{m-1}$.
Then $\sqcap_f^{min}(\mathbf{T_{m-1}})\leq g(m-1,k(k-1))$, otherwise the assumptions are already broken.
Additionally, if for any $T\in \mathbf{T}_{m-1}$ our $x$ was connected to $c$ variables $y_1,\dots,y_c$ such that there was no relation between them, then $\sqcap_f^{min}(\mathbf{T}_{m-1})\leq g(m-1,k(k-1)) - 2^{c-1}$.
Otherwise we could take e.g. the relations used for $x<_Ty_1$ and $x>_Ty_1$, and add these to $y_1<_Ty_2$ and $y_1>_Ty_2$ respectively, and so on.   
Effectively this means that in the step from $\mathbf{T}_{m-1}$ to $\mathbf{T}$, we can never surpass the case of connecting $x$ to $c$ other variables which all are in full relation to each other, giving $(c+1)g(m-1,k(k-1))$.
But as this is true for all $m$ variables and as $m$ is the smallest such value, they must all be fully connected to each other, which requires $m(m-1)$ relations.
But as we only have $k(k-1)$ possible relations, we must have $k=m$, which contradicts the original assumption.
Alternatively $m$ could be partitioned into sets of size $k(k-1)$, which breaks both the assumption that $m$ is the smallest value such that $g(m,k(k-1))>k!^{m/k}$, nor would $g(m,k(k-1))>k!^{m/k}$ hold.
Hence, our maximum is $k!^{(n-k(k-1))/k}$, for which we get the optimum $(cn/\log n)^{n/2}$, e.g. using Lambert's W function (see \cite{LambertW} for an overview), for some constant $c$.
\end{proof}

We combine everything and give the first $o(n)^n$ result for $\csp(\rcci{8}^U)$ (where $c \geq 0$ stems from Lemma~\ref{lem:PupperBound}).

\begin{theorem}
    An arbitrary $\csp(\rcci{8}^U)$ instance $I=(V,C)$ with $n=|V|$ is solvable in $\bigoh^*((cn/\log n)^{n})$ time and $\bigoh^*((cn/\log n)^{n/2})$ space.
\end{theorem}
\begin{proof}
By Lemma~\ref{lem:PoTCorrect} we have a correct algorithm in $R_f(S)$, by Lemma~\ref{lem:PupperBound} we know that for all $S\subseteq V$ then $|R_f(S)|\in \bigoh^*((cn/\log n)^{n/2})$, using a DP approach to calculate $R_f(S)$ for all $S\subseteq V$ can be done in $3^n$-steps by standard methods, and lastly, calculating $\sqcap_f^{min}$ in every step of the recurrence function/DP is in $\bigoh^*((2^{|S|}|R_f(S)|)^2)$ time by Lemma~\ref{lem:PCleanComplexity}.
Combining all this gives that $I$ is solvable in $\bigoh^*((cn/\log n)^{n})$ time and $\bigoh^*((cn/\log n)^{n/2})$ space.
\end{proof}

While the DP part of the algorithm has been reasonably hidden until this last theorem (or until the recurrence relation for the astute reader) the concept has been very relevant for why we want to minimize $\sqcap_f^{min}(\mathbf{T})$ and why focusing on IPs are relevant in the first place.

\section{Concluding Remarks} \label{sec:summary}

Our work opens up mathematically and algorithmically motivated questions. First, non-redundancy of finite-domain CSPs has recently shown to coincide with {\em $\epsilon$-sparsifiability}~\cite{brakensiek2024}. Is this possible for qualitative reasoning problems (often {\em $\omega$-categorical}) and other types of infinite-domain CSPs? This would allow one to  sparsify instances via coding theory merely by knowing $\nrd_{\Gamma}$.
Second, are there any other (NP-hard) fragments of qualitative reasoning problems solvable in single-exponential time? Here, one can get reasonable candidates by selecting a maximal tractable class and adding an arbitrary relation which makes the problem hard. For example, can we find $2^{\bigoh(n)}$ fragments of e.g. the {\em cardinal direction calculus} or the {\em star calculus} in this way? In the other direction, one may conjecture that NP-hard qualitative reasoning problems are solvable in $o(n)^n$ time but not substantially faster. Can this be formally established under e.g.\ the  {\em exponential-time hypothesis}~\cite{impagliazzo2001}?

\section*{Acknowledgments}

We thank the anonymous reviewers for their helpful suggestions on how to improve the presentation. The first author is partially supported by the Swedish research council under grant VR-2022-03214 and VR-2025-04487.

\bibliographystyle{apacite}
\bibliography{references}
\newpage

\end{document}